\documentclass[preprint,12pt]{elsarticle}

\usepackage{algpseudocodex}
\usepackage{algorithm}
\algrenewcommand\algorithmicrequire{\textbf{Input:}}
\algrenewcommand\algorithmicensure{\textbf{Output:}}

\usepackage{amsthm}
\usepackage{amsmath,amssymb}
\usepackage{graphicx}
\usepackage{booktabs}
\usepackage{subcaption}
\usepackage{amsfonts}
\usepackage{xcolor}
\usepackage{longtable}
\usepackage{float}
\usepackage[T1]{fontenc}
\usepackage[colorlinks=true, allcolors=blue, breaklinks=true]{hyperref}
\usepackage[title]{appendix}

\DeclareMathOperator\supp{supp}

\newtheorem{theorem}{Theorem}[section]
\theoremstyle{plain}

\newtheorem{lemma}[theorem]{Lemma}

\theoremstyle{definition}
\newtheorem{definition}[theorem]{Definition}

\journal{}

\begin{document}

\begin{frontmatter}

%% Title, authors and addresses

%% use the tnoteref command within \title for footnotes;
%% use the tnotetext command for theassociated footnote;
%% use the fnref command within \author or \affiliation for footnotes;
%% use the fntext command for theassociated footnote;
%% use the corref command within \author for corresponding author footnotes;
%% use the cortext command for theassociated footnote;
%% use the ead command for the email address,
%% and the form \ead[url] for the home page:
%% \title{Title\tnoteref{label1}}
%% \tnotetext[label1]{}
%% \author{Name\corref{cor1}\fnref{label2}}
%% \ead{email address}
%% \ead[url]{home page}
%% \fntext[label2]{}
%% \cortext[cor1]{}
%% \affiliation{organization={},
%%             addressline={},
%%             city={},
%%             postcode={},
%%             state={},
%%             country={}}
%% \fntext[label3]{}

\title{An orderly algorithm for generation of Condorcet Domains}
\date{}
% use optional labels to link authors explicitly to addresses:

\author[ic]{Bei Zhou}
% \ead{bei.zhou@imperial.ac.uk}
\affiliation[ic]{organization={Imperial College London},
                 country={United Kingdom}}

\author[umea]{Klas Markström\corref{cor1}}
\cortext[cor1]{Corresponding author}
\ead{klas.markstrom@umu.se}
\affiliation[umea]{organization={Ume\aa\ University},%Department and Organization
                   country={Sweden}}

%% Abstract
\begin{abstract}
    Condorcet domains are fundamental objects in the theory of majority voting; they are sets of linear orders with the property that if every voter picks a linear order from this set, assuming that the number of voters is odd,  and alternatives are ranked according to the pairwise majority ranking, then the result is a linear order on the set of all alternatives.

    In this paper we present an efficient orderly algorithm for the generation of all non-isomorphic maximal Condorcet domains on $n$ alternatives. The algorithm can be adapted to generate domains from various important subclasses of Condorcet domains. We use an example implementation to extend existing enumerations of domains from several such subclasses and make both data and the implementation publicly available.
\end{abstract}

% %%Graphical abstract
% \begin{graphicalabstract}
% %\includegraphics{grabs}
% \end{graphicalabstract}

% %%Research highlights
% \begin{highlights}
% \item Research highlight 1
% \item Research highlight 2
% \end{highlights}

%% Keywords
\begin{keyword}
%% keywords here, in the form: keyword \sep keyword
Condorcet domains \sep orderly generation \sep voting theory
%% PACS codes here, in the form: \PACS code \sep code

%% MSC codes here, in the form: \MSC code \sep code
%% or \MSC[2008] code \sep code (2000 is the default)

\end{keyword}

\end{frontmatter}

%-----------------------------------------------------------------
\section{Introduction}
A Condorcet domain is a set of linear orders on a set of alternatives such that if an odd number of voters select a linear order from this set to represent their preferences, the pairwise majority relation among these alternatives leads to a linear order.   Condorcet domains have been actively studied since the 1950's, motivated by the works of Black \cite{black} and Arrow \cite{Arrow1951}, but the earliest mathematical work on this topic was done by Marquis de Condorcet in 1785 \cite{Con}. Work on Condorcet domains has mainly been conducted by economists and social choice theorists \cite{puppe2024maximal}, but algorithmic aspects are increasingly studied by computer scientists \cite{elkind2022} and mathematicians have investigated their combinatorial structure \cite{monjardet2005social}.    

The early works provided some constructions of well-behaved Condorcet domains for general $n$, the number of alternatives. Black \cite{black} introduced the single-peaked domains, which correspond to an opinion landscape based on an underlying one-dimensional axis, e.g. signifying a left-right scale. Black's domains have size $2^{n-1}$ and are \emph{maximal}, meaning that they are not contained in any larger Condorcet domains.  Arrow \cite{arrow63} introduced a larger class of domains, which includes Black's, and these domains also turn out to have size at most $2^{n-1}$ (Recently, Slinko \cite{slinko2019condorcet} showed that all maximal members of this class have size $2^{n-1}$). Motivated by these and other examples Johnson \cite{Johnson78} conjectured that Condorcet domains cannot be larger than the single-peaked domains. Six years later Johnson and Abello \cite{MR763988} disproved this conjecture, though only by an additive constant. Fishburn \cite{fishburn1997acyclic} both constructed the well-structured domains which are now called Fishburn domains, which are a polynomial factor larger than $2^n$, and gave a recursively defined class of domains of size $\mathcal{O}(a^n)$, for an $a>2$.  On the other hand, Raz \cite{Raz} proved that the maximum size is bounded by $\mathcal{O}(c^n)$, for some $c>0$. 

Over the years many other conjectures and open problems on the structure of Condorcet domains were posed, often motivated by the behaviour of the known handcrafted examples for small $n$.  More recently, Dittrich \cite{Dittrich2018} enumerated all maximal non-isomorphic Condorcet domains on $n\leq 5$ alternatives and investigated their properties.  Here two domains are considered isomorphic if they differ only by a relabeling of the alternatives.  This enumeration was significantly extended in \cite{akello2025condorcet}, where all maximal domains on $n\leq 7$ alternatives were enumerated, using a new algorithm. The latter paper also investigated the properties of the constructed domains and found counterexamples to several conjectures as well as unexpected answers to some open problems.   In \cite{leedham2024largest} a variant on the same algorithm was used to determine the maximum size of a Condorcet domain for $n=8$ and it was found to differ strongly in structure from the maximum domains for smaller $n$.

%----------------------------
\subsection{Our contribution}
As the computational investigations mentioned above have shown, the study of Condorcet domains has gained much from the non-intuitive examples found by complete enumerations.  However, the algorithms used so far all have a common drawback, despite being very different algorithms, namely that they generate a huge number of isomorphic copies of each domain. Thus the generation step must be followed by a separate isomorphism reduction step in order to get just one member from each isomorphism class.  This two-step process becomes less and less efficient for larger $n$, where the size of each isomorphism class becomes larger.  The existing search algorithms work with actual domains during the search, something which works well for  small $n$ but becomes very memory consuming  for larger $n$, as the domains become larger.

The purpose of this paper is to describe an orderly algorithm for generating Condorcet domains, producing only one member of each isomorphism class as output. For a general discussion of orderly algorithms we recommend \cite{MR2192256}.  The algorithm drastically reduces the size of the underlying search tree and eliminates the need for a separate isomorphism reduction after the search. The algorithm can easily be adapted to produce only members of some  specific classes of domains, such as Arrow's single-peaked domains \cite{arrow63}, or the generalised  Fishburn domains \cite{Karpov2023,slinko2024family}.

The main difference between our algorithm and the existing ones is that instead of working with domains as sets of linear orders, we work with sets of never conditions.  A classic theorem of Sen \cite{Sen1966} shows that a Condorcet domain can be represented by a set of constraints now known as never conditions (these are defined in the next section).  Isomorphism of domains can be translated into isomorphism of sets of never conditions and our algorithm enumerates all isomorphism classes of such sets.   This approach substantially reduces memory usage as storing a set of never conditions requires at most $\mathcal{O}(n^3)$ bits, while a  full domain typically has size exponential in $n$.  There are some pitfalls here. As shown in \cite{Guttman2006} it is in general an NP-complete problem to determine if a set of never conditions corresponds to a non-empty domain, so a naive approach might generate sets which correspond to the empty domain. However, as noticed in \cite{akello2025condorcet} all Condorcet domains on the set of alternatives $[n] $are isomorphic to a \emph{unitary} domain, i.e. one which contains the standard linear order on $[n]$, and only some never conditions correspond to unitary domains. By working only with those never conditions we still generate all isomorphism classes  and avoid making an NP-complete test for each output from the search. Since we work with a smaller set of never conditions this also reduces the branching factor of the search tree, leading to another speed up.

As a demonstration of our algorithm, we enumerate all maximal domains in three restricted classes for $n=8$, thereby going a step past the results in \cite{akello2025condorcet}, and for one class we give data for $n=9, 10$ as well. For larger $n$ the number of domains becomes very large and targeted searches make more sense. 

%--------
\subsection{Organisation of the paper}
In section~\ref{sec:not} we give additional notation and background material. In Section~\ref{sec:Iso} we discuss the isomorphism of Condorcet domains and their connection to never conditions in more detail. In Section~\ref{sec:algo} our enumeration algorithm is described and in Section~\ref{sec:res} we demonstrate our implementation by enumerating three different domain classes for $n=8$, and in one case larger $n$ as well.

%-----------------------------------------------------------------
\section{Notation and definitions}\label{sec:not}

Let $X_n=[n]=\{1, \dots, n\}$ be the set of alternatives. Let $\mathcal{L}(X_n)$ be the set of all linear orders over $X_n$.  Each voter, sometimes called an agent, $i \in X$ has a preference order $P_i\in\mathcal{L}(X_n)$. For brevity, we will write preference orders as strings, e.g. $12 \dots n$ means $1$ is the most preferred alternative, $n$ is the least. 

A set of linear orders $\mathcal{D}\subseteq \mathcal{L}(X_n)$ is called a \emph{domain}. A domain $\mathcal{D}$ is a \emph{Condorcet domain} if for an odd number of voters with preferences from the domain, the pairwise majority between the alternatives leads to a transitive ordering of all alternatives. There are many other, equivalent, ways to define a Condorcet domain \cite{monjardet2009acyclic} and we shall later introduce Sen's version based on never conditions.

A Condorcet domain is \emph{unitary} if it contains the standard linear order $\alpha=123\ldots n$. As noted in \cite{akello2025condorcet}, all Condorcet domains are isomorphic to  some unitary domain.

A Condorcet domain $\mathcal{D}$ is \emph{maximal} if every Condorcet domain $\mathcal{D}'\supseteq \mathcal{D}$ (on the same set of alternatives) coincides with $\mathcal{D}$. 

The \emph{restriction} of a domain $\mathcal{D}$ to a subset $A\subset X_n$ is the subset of linear orders from $\mathcal{L}(A)$ obtained by restricting each linear order from $\mathcal{D}$ to $A$. A Condorcet domain $\mathcal{D}$ is \emph{copious} \cite{slinko2019condorcet} if the restriction of $\mathcal{D}$ to any triple of alternatives has size 4. 

A \emph{triple} refers to a tuple of three distinct alternatives in ascending order, according to the standard order $\alpha$.

%-----------------------------
\subsection{Never conditions}
\cite{Sen1966} proved that a domain is a Condorcet domain if and only if the restriction of the domain to any triple of alternatives $\{a,b,c\}$ satisfies a never condition.

In Sen's notation, a never condition on a triple $\{a,b,c\}$ can be of three forms $xNb$, $xNm$, $xNt$, referred to as a never bottom, a never middle, and a never top condition respectively. Here $x$ is an alternative from the triple and  $xNb$, $xNm$, and $xNt$ means that $x$ is not ranked last, second, or first respectively in the restricted domain. 

We say that a linear order $\gamma$ is \emph{compatible} with a given never condition if the restriction of $\gamma$ to that triple satisfies the never condition.  

Fishburn noted that for unitary domains, never conditions can instead be described as $iNj$, $i,j\in [3]$. Here $iNj$ applied to a triple of alternatives means that $i^{th}$ alternative from the triple, according to the standard order,  does not get ranked in $j^{th}$ place when an order from the domain is restricted to this triple. For example, restriction of  $abc$ to the triple $\{a,b,c\}\in[n]$,  satisfies the never conditions $1N2,1N3,2N1,2N3,3N1,3N2$, but violates never conditions $1N1,2N2,3N3$.  Hence every unitary domain will satisfy some never conditions of the forms $1N2,1N3,2N1,2N3,3N1,3N2$ for every triple.

A set of never conditions which includes at least one never condition for every triple is called a \emph{complete} set of never conditions. Sen's theorem can be formulated as saying that a domain is a Condorcet domain if and only if it satisfies a complete set of never conditions.

A Condorcet domain $\mathcal{D}$ is a \emph{full} Condorcet domain if it contains every linear order which is compatible with the same never conditions as $\mathcal{D}$.  A full Condorcet domain can equivalently be represented as either a set of linear orders or a set of never conditions.
 
A set $S$ of never conditions \emph{implies} a never condition $iNj$, which is not part of $S$, on a triple $\{a,b,c\}$ if every linear order which is compatible with $S$ also satisfies the never condition $iNj$ on that triple.

Note that a maximal Condorcet domain is full, but a full Condorcet domain need not be maximal.  Similarly, if the never conditions satisfied by a Condorcet domain $\mathcal{D}_1$ are a strict subset of those satisfied by $\mathcal{D}_2$ then $\mathcal{D}_2 \subset \mathcal{D}_1$, and $\mathcal{D}_2$ is not maximal. Together this leads to the following lemma, which we  consider folklore but have not found a full proof of in the existing literature.

\begin{lemma}
    A full Condorcet domain is maximal if and only if any triple which satisfies more than one never condition only satisfies implied never conditions.    
\end{lemma}
\begin{proof}
    If some triple $T$ in $\mathcal{D}$ satisfies more than one non-implied never condition then let $\mathcal{D}'$ be a full domain which satisfies only one of those never conditions on $T$, and the same conditions as $\mathcal{D}$ on every other triple.  Then both $\mathcal{D}$ and $\mathcal{D}'$ satisfy a complete set of never conditions and $D$ is a strict subset of $\mathcal{D}'$. So, $\mathcal{D}$ is not maximal. 

    Assume that $\mathcal{D}$ satisfies both a non-implied never condition and an implied never condition on some triple $T$.  Let $\mathcal{D}'$ be the domain which satisfies the same never conditions as $\mathcal{D}$, except the non-implied condition on $T$.   Then both $\mathcal{D}$ and $\mathcal{D}'$ satisfy a complete set of never conditions and $\mathcal{D}$ is a strict subset of $\mathcal{D}'$. So, $\mathcal{D}$ is not maximal.

    Assume that $D$ is such that the triples which satisfy more than one condition only satisfy implied never conditions.  If $\mathcal{D}$ is not maximal then there is some domain $\mathcal{D}'$ such that $\mathcal{D}$ is a strict subset of $\mathcal{D}'$. This does in turn mean that the set of never conditions for $\mathcal{D}'$ is a strict subset of those for $\mathcal{D}$. However, $\mathcal{D}'$ must satisfy the same never conditions as $\mathcal{D}$ on the triples with only one never condition, and those imply all never conditions on triples where $\mathcal{D}$ satisfy more than one condition. So, $\mathcal{D}'$ cannot be a strict superset of $\mathcal{D}$  and $\mathcal{D}$ is maximal.
\end{proof}

We will use $\mathcal{N}$ to denote a set of triples and their never conditions, $\mathcal{N}(\mathcal{D})$ to refer to the $\mathcal{N}$ that leads to the full domain $\mathcal{D}$, and $\mathcal{D} (\mathcal{N})$ to refer to the maximal domain $\mathcal{D}$ that satisfies $\mathcal{N}$. 

The \emph{support}  $\supp{\mathcal{(N)}}$ of $\mathcal{N}$  is the set of alternatives appearing in the triples of $\mathcal{N}$. 

$\mathcal{U}(\mathcal{N})$ denotes the set of triples to which $\mathcal{N}$ does not assign a never condition. Analogously, we use $\mathcal{A}(\mathcal{N})$ to denote the set of assigned triples in $\mathcal{N}$ and their never conditions.

%---------------------------------------------------------
\subsection{Lexicographic orders}
Assume that we have a given linear order on the set of triples from our set of alternatives. Let us also assume that we have another linear order on the set of Fishburn type never conditions. Given the latter order we can represent each never condition by an integer between 1 and 6. 

A set of never conditions $\mathcal{N}$ can now be represented as a string $S(\mathcal{N})$ where each never condition is represented by an integer and they are listed in the same order as their corresponding triple in $\mathcal{N}$.  In order to get strings of a fixed length, we can assign a triple with no never condition the integer 0, thereby making such triples ranked lower than all never conditions.

Using this string representation, we can compare two sets of never conditions $\mathcal{N}_1$ and $\mathcal{N}_2$ lexicographically.  We say that $\mathcal{N}_1$ is lexicographically larger than $\mathcal{N}_2$ if $S(\mathcal{N}_1)$ is lexicographically larger than $S(\mathcal{N}_2)$. 

The lexicographic order induced by the two linear orders on triples and never conditions respectively is a linear order on the set of sets of never conditions.

%---------------------------------------------------------
\section{Isomorphism and never conditions}\label{sec:Iso} 
The standard definition of isomorphism for (Condorcet) domains states that $\mathcal{D}_1$ and $\mathcal{D}_2$ are \emph{isomorphic} if $\mathcal{D}_1$ can be obtained by relabeling the alternatives in $\mathcal{D}_2$.

\begin{definition}[Domain Isomorphism]
    \label{def:cd_iso}
    Let $\mathcal{D}_1$ and $\mathcal{D}_2$ be two domains of equal cardinality on sets of alternatives $X_1$ and $X_2$, respectively. 

    We say that domains $\mathcal{D}_1$ and $\mathcal{D}_2$ are {\em isomorphic} if there is a bijection $\pi\colon X_1\to X_2$  such that for each linear order $v \in \mathcal{D}_1$  we have $\pi(v) \in \mathcal{D}_2$, where ${\pi(v)}$ is the order $v$ with the names of the alternatives permuted according to $\pi$. 
\end{definition}

Our algorithm will work with sets of never conditions and we shall need a formulation of isomorphism in terms of these as well.  We say that two sets of never conditions are isomorphic if they define isomorphic full domains. This can also be expressed directly in terms of never conditions.

\begin{definition}
    Given a never condition $xNy$ on a triple $T_1$ and a bijection $\psi$ from $T_1$ to $T_2=\{\psi(a), \psi(b), \psi(c)\}$, we say that the never condition induced on $T_2$ is $(\texttt{index}(T_2, \psi(T_1[x-1]))+1)Ny$ , where $\texttt{index}(T, i)$ denotes the position of $i$ in a triple $T$. 

    Given a set of never conditions $\mathcal{N}_1$ and a bijection $\psi$ from $\supp(\mathcal{N})$ to $S$ we define $\psi(\mathcal{N})$ to be the set of all never conditions induced by the never conditions in $\mathcal{N}$ on the triples from $S$.
    
\end{definition}

\begin{definition}[Isomorphism on Sets of Never Conditions]
    \label{def:nc_iso}
    Let $\mathcal{N}_1$ and $\mathcal{N}_2$ be two sets of never conditions on sets of alternatives $X_1$ and $X_2$, respectively. 
    
    We say that $\mathcal{N}_1$ and $\mathcal{N}_2$ are isomorphic if they have the same cardinality and there is a bijection $\psi\colon X_1\to X_2$  such that for any triple $T_1={\{a, b, c\}}$ and their never condition $xNy$ in $\mathcal{N}_1$, the triple $T_2=\{\psi(a), \psi(b), \psi(c)\}$ has the never condition induced by $\psi$ in $\mathcal{N}_2$.
    
\end{definition}
It is straightforward to check that two non-empty Condorcet domains $\mathcal{D}_1$ and $\mathcal{D}_2$ are isomorphic if and only if their sets of never conditions   $\mathcal{N}(\mathcal{D}_1)$ and  $\mathcal{N}(\mathcal{D}_2)$  are isomorphic in the sense given here.

Trivially, any set of never conditions $\mathcal{N}$ is isomorphic to the set which is lexicographically maximal among the sets which are isomorphic to $\mathcal{N}$.

%---------------------------------------------------------
\section{The orderly generation algorithm}\label{sec:algo} 
Previous algorithms have focused on generating all maximal Condorcet domains rather than all Condorcet domains. This is a sensible choice both since the number of maximal domains is much smaller than the total number of domains, and  every domain  is a subset of some maximal domain.  The basic version of our algorithm will generate the larger class of all \emph{full} domains, which includes all maximal domains. This basic algorithm can be combined with the maximality test used in the algorithm from \cite{akello2025condorcet} in order to only generate maximal domains, at the cost of additional overhead in the search.

The earlier algorithms used in \cite{Dittrich2018} and \cite{akello2025condorcet} worked with sets of linear orders during the search. As $n$ increases this becomes very costly in terms of memory, so in order to save memory our generation algorithm instead works on sets of never conditions.  For each isomorphism class of complete sets of never conditions, only the lexicographically maximal one is retained. The corresponding Condorcet domain can be found by a separate program after the search, if the domain is needed as an explicit set of linear orders.

The results in \cite{akello2025condorcet} show that already for $n=7$ the total number of maximal Condorcet domains is very large, and for $n=8$ even storing all the domains would be difficult.  With that in mind our algorithm is able to generate all full domains from a subclass defined by a set $\mathcal{R}$ of allowed never conditions.  That is, given $\mathcal{R}$ as input we only generate complete sets of never conditions where every triple has a never condition from $\mathcal{R}$. This makes it possible for the user to focus on well-studied domain classes such as the Arrow's single-peaked domains or the peak-pit  domains.

%----------------------------
\subsection{The generation algorithm}
We start with four pieces of input: the number of alternatives $n$, a set  $\mathcal{R}$ of allowed never conditions, a linear order $\prec_1$ on the set of triples from the set of alternatives, and a linear order $\prec_2$ on the set of never conditions\footnote{ 
In our implementation we use the co-lexicographic order on the set of triples and a fixed order on the set of all never conditions. So the only user inputs are the number of alternatives and the number of alternatives.}.

Our algorithm performs a depth-first search, at each step assigning an allowed never condition to the next triple according to the triple ordering $\prec_1$. 
After a triple has been assigned a never condition, the algorithm makes a lexicographic comparison of the corresponding set $\mathcal{N}$ of never conditions with those of its isomorphic forms which use only never conditions from $\mathcal{R}$. If one of those is lexicographically larger than $N$ then $N$ is rejected and the next never condition according to $\prec_2$ is  tried, when all never conditions in $\mathcal{R}$ have been tried the search backs up to the previous level. If a never condition is accepted the search moves on to the next triple. If there are no unassigned triples, the algorithm outputs a complete set of never conditions and backs up to the previous level of  the search tree. 

This algorithm is complete in the sense that it generates one representative for each isomorphism class of complete sets of never conditions.  The exact representative chosen depends on both $\prec_1$ and $\prec_2$. When the set $\mathcal{R}$  defines a class of Condorcet domains  which are always copious, the algorithm will generate only the sets of never conditions corresponding to maximal domains, without the need for a separate maximality test. 

The main algorithm is illustrated in Figure \ref{alg1}.  As we can see, the algorithm is quite simple as long as we have a test for lexicographic maximality of non-complete sets of never conditions. This test is the core of our algorithm, and leads to several non-trivial issues.

\begin{figure}

\begin{algorithm}[H]
\caption{Generation of complete sets of never conditions.}\label{alg:dfs_generate}
\begin{algorithmic}
    \Require{The number of alternatives $n$ where $n \geq 3$ \\
             a set of never conditions $\mathcal{R}$}
    \Ensure{A complete collection of sets of non-isomorphic, lexicographically maximal never conditions}

    \State Create an empty list $\mathcal{L}$ to store complete sets of never conditions
   
    \State Compute $\Call{DFS}{\{\}}$, where $\Call{DFS}$ is as follows:
    \\
    \Function{DFS}{$\mathcal{N}$}
        \If {$\mathcal{U}(\mathcal{N})$ is empty}
            \State add $\mathcal{N}$ to $\mathcal{L}$
        \Else
            \State Select the first unassigned triple $T$ from $\mathcal{U}(\mathcal{N})$
            \For {each never condition $R$ in $\mathcal{R}$}
                \State Assign $R$ to $T$ for $\mathcal{N}$
                \If {$\mathcal{N}$ is lexicographically maximal} \Comment{Checked via Algorithm~\ref{alg:verify}}
                    \State $\Call{DFS}{\mathcal{N}}$
                \EndIf
            \EndFor
        \EndIf
    \EndFunction

    \State $\Call{DFS}{\mathcal{N}}$

    \State return $\mathcal{L}$
\end{algorithmic}
\end{algorithm}    
    \caption{The main generation algorithm}
    \label{alg1}
\end{figure}

%----------------------------
\subsection{Testing lexicographic maximality}
In order for the main algorithm to work we must be able to check if a set of never conditions $\mathcal{N}$ is lexicographically maximal, within the allowed such sets. That is, we should only compare with sets using the allowed never conditions in $\mathcal{R}$.  If the sets are complete this is straightforward, we simply apply each permutation of the set of alternatives to $\mathcal{N}$, check if all the induced never conditions belong to $\mathcal{R}$, and then make a lexicographic comparison of the corresponding strings.  In order for the algorithm to provide a non-trivial speed up we must also be able to compare non-complete sets of never conditions, since that is what allows us to remove branches from the search tree. For non-complete sets several hurdles appear, due to the presence of unassigned triples. 

First, a permutation might map the set of assigned triples to itself, but  while doing that  it might map some unassigned triple $T$ in such a way that any allowed never condition on $T$ is mapped to a never condition not in $\mathcal{R}$.  Such a permutation does not correspond to a comparison of complete sets of never conditions from $\mathcal{R}$, and so must be discarded.,

Second, a permutation might map an unassigned triple to an assigned one. In doing so it might also map never conditions in $\mathcal{R}$ to never conditions which are lexicographically lower than those actually used by the search for this assigned triple. This makes it impossible to know if the transformed set can validly be compared to the original one.

In order to avoid these problems we restrict our maximality test to a partial one for non-complete sets.  When transforming our input set we only consider permutations  which induce an allowed never condition for assigned triples, map every allowed never condition on unassigned triple to an allowed one, and do not map unassigned triples to assigned one.   This rejects more permutations than necessary, but it keeps the tests simple and avoids overhead costs in the search.  

The test algorithm is illustrated in Figure \ref{alg2}.

As a simple speed up we also only consider permutations which act as the identity outside the support of the current set of never conditions.
In our implementation we have chosen to work with the co-lexicographic order on the set of triples since this keeps the support of $\mathcal{N}$ as small as possible, relative to the number of assigned triples.  

\begin{figure}
\begin{algorithm}[H]
\caption{Partial lexicographical maximality test}\label{alg:verify}
\begin{algorithmic}
    \Require{A set never conditions $\mathcal{N}$ on the set of alternatives $[n]$, where $n \geq 3$ \\
             A set of allowed never conditions $\mathcal{R}$} 
    \Ensure{TRUE  if the input $\mathcal{N}$ can be lexicographically maximal, FALSE otherwise}

    \State Let $\mathcal{P}$ be the set of permutations of $[m]$, where $m$ is the largest alternative in $\supp(\mathcal{N})$.
    
    \For {each permutation $P$ in $\mathcal{P}$}
        \State Create a list of unassigned triples $\widetilde{\mathcal{N}}$ in the same order as they present in $\mathcal{N}$. 
        \For {(each triple $T$, and its never condition $R$) in $\mathcal{A}(\mathcal{N})$}
            \State transform ($T$, $R$) to ($g(T), g(R)$)
            \If {$g(R)$ is not in $\mathcal{R}$}
                \State skip to the next permutation.
            \EndIf 
            assign $g(R)$ to $g(T)$ in $\widetilde{\mathcal{N}}$
        \EndFor
        \For {each triple $T$ in $\mathcal{U}(\mathcal{N})$}
            \For {each never condition $R$ in $\mathcal{R}$}
                \State transform ($T$, $R$) to ($g(T), g(R)$)
                \If {$g(R)$ is not in $\mathcal{R}$}
                    \State skip to the next permutation.
                \EndIf 
            \EndFor

            \If {assignment order is not identical to triple order}
                \If {$g(T)$ has a never condition in $\mathcal{N}$}
                    \State skip to the next permutation. 
                \EndIf 
            \EndIf
        \EndFor

        \If {$\widetilde{\mathcal{N}}$ is lexicographically larger than $\mathcal{N}$}
            \State return False
        \EndIf
    \EndFor
    \State return True
\end{algorithmic}
\end{algorithm}
    \caption{The lexicographic maximality test.}
    \label{alg2}
\end{figure}

%---------------------------------------------------------
\section{Computational results}\label{sec:res}
The implementation of our search algorithm uses the Condorcet Domain Library (CDL) \cite{zhou2024cdl} which provides a range of functions for Condorcet domain related computations, such as constructing the domain from a list of triples and their assigned rules. Its robustness and reliability have been attested by recent studies \cite{karpov2025coherent, zhou2025efficient, akello2025condorcet, karpov2024local, markstrom2024arrow} utilizing it to generate or verify domains. Notably, a fast implementation of our isomorphic reduction function has been added to this library and is now publicly available.

In \cite{akello2025condorcet} the maximal Condorcet domains on at most 7 alternatives were generated and classified up to isomorphism. For $n=8$ the total number of maximal domains is too large to generate, even after isomorphism reduction.    As a demonstration of our algorithm we have extended the generation of maximal domains for several subclasses of domains to larger values of $n$.   Our example classes are the three defined by sets of never conditions from the pairs $2N3-2N1$, $1N3-3N1$, and $1N3-2N1$ respectively. In each case these are peak-pit domains.  We have here used the basic version of our algorithm which generates all full domains from the given subclass, rather than just the maximal ones. However, it turns out that all generated domains are  copious and hence also maximal domains.

In Tables \ref{tab:iso_2n32n1_8} to \ref{tab:iso_1n32n1_8} we display the number of maximal domains from each subclass according to their sizes.  The same statistics for domains on 9 and 10 alternatives can be found in the supplementary material.

\begin{table}[H]
\caption{The number of non-isomorphic copious $2N3 - 2N1$ domains and their sizes for $n=8$. Domain sizes are listed in ascending order. }\label{tab:iso_2n32n1_8}
\small
\centering
\begin{tabular}{l|*{13}{c}}
\toprule

Size & 29 & 30 & 31 & 32 & 33 & 34 & 35 & 36 & 37 & 38 & 39 & 40 & 41  \\
Count & 2 & 4 & 6 & 12 & 46 & 42 & 115 & 155 & 183 & 289 & 406 & 511 & 609  \\
\hline
Size & 42 & 43 & 44 & 45 & 46 & 47 & 48 & 49 & 50 & 51 & 52 & 53 & 54  \\
Count & 852 & 987 & 1226 & 1506 & 1674 & 1920 & 2100 & 2298 & 2494 & 2634 & 2886 & 3135 & 3320  \\
\hline
Size & 55 & 56 & 57 & 58 & 59 & 60 & 61 & 62 & 63 & 64 & 65 & 66 & 67  \\
Count & 3444 & 3399 & 3522 & 3593 & 3746 & 3939 & 3837 & 3923 & 3779 & 4015 & 3913 & 4055 & 3984  \\
\hline
Size & 68 & 69 & 70 & 71 & 72 & 73 & 74 & 75 & 76 & 77 & 78 & 79 & 80  \\
Count & 4392 & 4191 & 4514 & 4331 & 4681 & 4405 & 4657 & 4497 & 4970 & 4568 & 4955 & 4893 & 5259  \\
\hline
Size & 81 & 82 & 83 & 84 & 85 & 86 & 87 & 88 & 89 & 90 & 91 & 92 & 93  \\
Count & 5171 & 5251 & 5158 & 5352 & 5034 & 5382 & 5070 & 5352 & 5135 & 5426 & 5211 & 5230 & 4747  \\
\hline
Size & 94 & 95 & 96 & 97 & 98 & 99 & 100 & 101 & 102 & 103 & 104 & 105 & 106  \\
Count & 5039 & 4741 & 5085 & 4475 & 4924 & 4331 & 4566 & 4021 & 4343 & 3779 & 4108 & 3503 & 3850  \\
\hline
Size & 107 & 108 & 109 & 110 & 111 & 112 & 113 & 114 & 115 & 116 & 117 & 118 & 119  \\
Count & 3338 & 3540 & 3195 & 3535 & 2980 & 3489 & 2878 & 3366 & 2860 & 3228 & 2474 & 3226 & 2577  \\
\hline
Size & 120 & 121 & 122 & 123 & 124 & 125 & 126 & 127 & 128 & 129 & 130 & 131 & 132  \\
Count & 3017 & 2361 & 2796 & 2283 & 2667 & 2231 & 2679 & 2029 & 2304 & 1833 & 2133 & 1683 & 1916  \\
\hline
Size & 133 & 134 & 135 & 136 & 137 & 138 & 139 & 140 & 141 & 142 & 143 & 144 & 145  \\
Count & 1574 & 1774 & 1417 & 1569 & 1203 & 1362 & 1078 & 1166 & 972 & 1089 & 852 & 970 & 752  \\
\hline
Size & 146 & 147 & 148 & 149 & 150 & 151 & 152 & 153 & 154 & 155 & 156 & 157 & 158  \\
Count & 899 & 690 & 844 & 530 & 745 & 503 & 706 & 486 & 678 & 422 & 593 & 399 & 567  \\
\hline
Size & 159 & 160 & 161 & 162 & 163 & 164 & 165 & 166 & 167 & 168 & 169 & 170 & 171  \\
Count & 365 & 488 & 352 & 443 & 267 & 369 & 241 & 305 & 191 & 280 & 150 & 173 & 102  \\
\hline
Size & 172 & 173 & 174 & 175 & 176 & 177 & 178 & 179 & 180 & 181 & 182 & 183 & 184  \\
Count & 145 & 76 & 106 & 56 & 83 & 35 & 84 & 36 & 67 & 22 & 62 & 31 & 49  \\
\hline
Size & 185 & 186 & 187 & 188 & 189 & 190 & 191 & 192 & 193 & 194 & 195 & 196 & 197  \\
Count & 33 & 55 & 30 & 31 & 25 & 44 & 19 & 28 & 19 & 29 & 5 & 16 & 15  \\
\hline
Size & 198 & 199 & 200 & 201 & 202 & 204 & 206 & 208 & 209 & 212 & 213 & 218 & 222  \\
Count & 2 & 4 & 2 & 1 & 4 & 2 & 2 & 1 & 1 & 3 & 2 & 1 & 1  \\

\bottomrule
\end{tabular}

\end{table}

% 1N3-3N1 $n=8$ Copious
% OrderedDict([(128, 61856), (129, 7131), (130, 11682), (131, 8500), (132, 13517), (133, 6335), (134, 15533), (135, 9405), (136, 13922), (137, 8907), (138, 10513), (139, 6950), (140, 10361), (141, 6018), (142, 7279), (143, 6532), (144, 6710), (145
% , 4588), (146, 6107), (147, 3878), (148, 4784), (149, 3779), (150, 4379), (151, 3067), (152, 4487), (153, 3262), (154, 4617), (155, 3425), (156, 4719), (157, 3657), (158, 3592), (159, 3218), (160, 3320), (161, 2829), (162, 2847), (163, 2476), (164, 2290), (165, 2001), (166, 1786), (167, 1786), (168, 1615), (169, 1154), (170, 1244), (171, 747), (172, 799), (173, 869), (174, 1206), (175, 588), (176, 938), (177, 500), (178, 975), (179, 523), (180, 665), (181, 549), (182, 779), (183, 399),
%  (184, 732), (185, 575), (186, 489), (187, 376), (188, 429), (189, 293), (190, 380), (191, 198), (192, 204), (193, 184), (194, 282), (195, 82), (196, 231), (197, 86), (198, 99), (199, 113), (200, 132), (201, 56), (202, 100), (203, 22), (204, 17)
% , (205, 12), (206, 2), (207, 4), (209, 6), (211, 6)])

\begin{table}[H]
\caption{The number of non-isomorphic copious $1N3-3N1$ domains and their sizes for $n=8$. Domain sizes are listed in ascending order. }\label{tab:iso_1n33n1_8}
\small
\centering
\begin{tabular}{l|*{12}{c}}
\toprule

Size & 128 & 129 & 130 & 131 & 132 & 133 & 134 & 135 & 136 & 137 & 138 & 139  \\
Count & 61856 & 7131 & 11682 & 8500 & 13517 & 6335 & 15533 & 9405 & 13922 & 8907 & 10513 & 6950  \\
\hline
Size & 140 & 141 & 142 & 143 & 144 & 145 & 146 & 147 & 148 & 149 & 150 & 151  \\
Count & 10361 & 6018 & 7279 & 6532 & 6710 & 4588 & 6107 & 3878 & 4784 & 3779 & 4379 & 3067  \\
\hline
Size & 152 & 153 & 154 & 155 & 156 & 157 & 158 & 159 & 160 & 161 & 162 & 163  \\
Count & 4487 & 3262 & 4617 & 3425 & 4719 & 3657 & 3592 & 3218 & 3320 & 2829 & 2847 & 2476  \\
\hline
Size & 164 & 165 & 166 & 167 & 168 & 169 & 170 & 171 & 172 & 173 & 174 & 175  \\
Count & 2290 & 2001 & 1786 & 1786 & 1615 & 1154 & 1244 & 747 & 799 & 869 & 1206 & 588  \\
\hline
Size & 176 & 177 & 178 & 179 & 180 & 181 & 182 & 183 & 184 & 185 & 186 & 187  \\
Count & 938 & 500 & 975 & 523 & 665 & 549 & 779 & 399 & 732 & 575 & 489 & 376  \\
\hline
Size & 188 & 189 & 190 & 191 & 192 & 193 & 194 & 195 & 196 & 197 & 198 & 199  \\
Count & 429 & 293 & 380 & 198 & 204 & 184 & 282 & 82 & 231 & 86 & 99 & 113  \\
\hline
Size & 200 & 201 & 202 & 203 & 204 & 205 & 206 & 207 & 209 & 211  \\
Count & 132 & 56 & 100 & 22 & 17 & 12 & 2 & 4 & 6 & 6  \\

\bottomrule
\end{tabular}

\end{table}

% 1N3-2N1 $n=8$ Copious
% OrderedDict([(44, 7), (45, 17), (46, 7), (47, 13), (48, 15), (49, 11), (50, 41), (51, 21), (52, 14), (53, 26), (54, 48), (55, 48), (56, 42), (57, 18), (58, 17), (59, 31), (60, 32), (61, 31), (62, 29), (63, 26), (64, 27), (65, 45), (66, 25), (67, 29), (68, 38), (69, 23), (70, 37), (71, 37), (72, 33), (73, 37), (74, 43), (75, 55), (76, 57), (77, 38), (78, 38), (79, 73), (80, 89), (81, 58), (82, 44), (83, 67), (84, 66), (85, 62), (86, 48), (87, 47), (88, 49), (89, 71), (90, 35), (91, 59), (92, 46), (93, 60), (94, 43), (95, 41), (96, 59), (97, 47), (98, 31), (99, 30), (100, 46), (101, 27), (102, 25), (103, 29), (104, 34), (105, 55), (106, 53), (107, 52), (108, 39), (109, 32), (110, 56), (111, 33), (112, 32), (113, 47), (114, 45), (115, 22), (116, 37), (117, 33), (118, 28), (119, 45), (120, 51), (121, 44), (122, 43), (123, 11), (124, 43), (125, 23), (126, 22), (127, 26), (128, 36), (129, 22), (130, 24), (131, 26), (132, 26), (133, 26), (134, 18), (135, 23), (136, 31), (137, 16), (138, 23), (139, 14), (140, 24), (141, 9), (142, 16), (143, 19), (144, 16), (145, 3), (146, 37), (147, 10), (148, 19), (149, 2), (150, 14), (151, 8), (152, 17), (153, 6), (154, 18), (155, 9), (156, 9), (157, 6), (158, 14), (159, 6), (160, 7), (161, 4), (162, 13), (163, 3), (164, 1), (165, 2), (166, 3), (167, 2), (168, 9), (170, 5), (171, 3), (172, 4), (174, 5), (175, 1), (176, 2), (178, 4), (180, 2), (182, 2), (184, 3), (186, 2), (190, 1), (194, 1)])

\begin{table}[H]
\caption{The number of non-isomorphic copious $1N3-2N1$ domains and their sizes for $n=8$. Domain sizes are listed in ascending order. }\label{tab:iso_1n32n1_8}
\small
\setlength{\tabcolsep}{5.5pt}
\centering
\begin{tabular}{l|*{16}{c}}
\toprule

Size & 44 & 45 & 46 & 47 & 48 & 49 & 50 & 51 & 52 & 53 & 54 & 55 & 56 & 57 & 58 & 59  \\
Count & 7 & 17 & 7 & 13 & 15 & 11 & 41 & 21 & 14 & 26 & 48 & 48 & 42 & 18 & 17 & 31  \\
\hline
Size & 60 & 61 & 62 & 63 & 64 & 65 & 66 & 67 & 68 & 69 & 70 & 71 & 72 & 73 & 74 & 75  \\
Count & 32 & 31 & 29 & 26 & 27 & 45 & 25 & 29 & 38 & 23 & 37 & 37 & 33 & 37 & 43 & 55  \\
\hline
Size & 76 & 77 & 78 & 79 & 80 & 81 & 82 & 83 & 84 & 85 & 86 & 87 & 88 & 89 & 90 & 91  \\
Count & 57 & 38 & 38 & 73 & 89 & 58 & 44 & 67 & 66 & 62 & 48 & 47 & 49 & 71 & 35 & 59  \\
\hline
Size & 92 & 93 & 94 & 95 & 96 & 97 & 98 & 99 & 100 & 101 & 102 & 103 & 104 & 105 & 106 & 107  \\
Count & 46 & 60 & 43 & 41 & 59 & 47 & 31 & 30 & 46 & 27 & 25 & 29 & 34 & 55 & 53 & 52  \\
\hline
Size & 108 & 109 & 110 & 111 & 112 & 113 & 114 & 115 & 116 & 117 & 118 & 119 & 120 & 121 & 122 & 123  \\
Count & 39 & 32 & 56 & 33 & 32 & 47 & 45 & 22 & 37 & 33 & 28 & 45 & 51 & 44 & 43 & 11  \\
\hline
Size & 124 & 125 & 126 & 127 & 128 & 129 & 130 & 131 & 132 & 133 & 134 & 135 & 136 & 137 & 138 & 139  \\
Count & 43 & 23 & 22 & 26 & 36 & 22 & 24 & 26 & 26 & 26 & 18 & 23 & 31 & 16 & 23 & 14  \\
\hline
Size & 140 & 141 & 142 & 143 & 144 & 145 & 146 & 147 & 148 & 149 & 150 & 151 & 152 & 153 & 154 & 155  \\
Count & 24 & 9 & 16 & 19 & 16 & 3 & 37 & 10 & 19 & 2 & 14 & 8 & 17 & 6 & 18 & 9  \\
\hline
Size & 156 & 157 & 158 & 159 & 160 & 161 & 162 & 163 & 164 & 165 & 166 & 167 & 168 & 170 & 171 & 172  \\
Count & 9 & 6 & 14 & 6 & 7 & 4 & 13 & 3 & 1 & 2 & 3 & 2 & 9 & 5 & 3 & 4  \\
\hline
Size & 174 & 175 & 176 & 178 & 180 & 182 & 184 & 186 & 190 & 194  \\
Count & 5 & 1 & 2 & 4 & 2 & 2 & 3 & 2 & 1 & 1  \\

\bottomrule
\end{tabular}

\end{table}

%---------------------------------------------------------
\section*{Data Availability}

All the data presented in this paper are available on the author's website. 

%---------------------------------------------------------
\section*{Acknowledgment}

This research was conducted using the resources provided by High Performance Computing Center North (HPC2N). This research also utilised Queen Mary's Apocrita HPC facility, supported by QMUL Research-IT.

\bibliographystyle{elsarticle-num} 
\bibliography{reference}
\end{document}